%% file: dto-main.tex
\pdfoutput=1

\documentclass{llncs}

\usepackage{graphicx}

\begin{document}

\pagestyle{headings}  

\mainmatter              

\title{Incremental Topological Ordering and Strong Component Maintenance}

\author{Bernhard Haeupler\inst{1} \and Siddhartha Sen\inst{1} \and Robert E. Tarjan\inst{1}\inst{,2}}

\institute{Princeton University, Princeton NJ 08544\\
\email{\{haeupler, sssix, ret\}@cs.princeton.edu}
\and
HP Laboratories, Palo Alto CA 94304 \\
\hspace{1pt} \\ 
\today}

\maketitle              

\input{abstract}
\input{intro}
\input{bi-search}

\input{c-search}
\input{impl}

\input{strong}

\input{open}

%
%
\bibliography{local}
\bibliographystyle{abbrv}

\end{document}

%% file: abstract.tex
\begin{abstract}

We present an on-line algorithm for maintaining a topological order of a directed
acyclic graph as arcs are added, and detecting a cycle when one is created.  Our
algorithm takes $O(m^{1/2})$ amortized time per arc, where $m$ is the total number of
arcs.  For sparse graphs, this bound improves the best previous bound by a
logarithmic factor and is tight to within a constant factor for a natural class of
algorithms that includes all the existing ones.  Our main insight is that the
bidirectional search method of previous algorithms does not require an ordered
search, but can be more general.  This allows us to avoid the use of heaps (priority
queues) entirely.  Instead, the deterministic version of our algorithm uses
(approximate) median-finding.  The randomized version of our algorithm avoids this
complication, making it very simple.  We extend our topological ordering algorithm to
give the first detailed algorithm for maintaining the strong components of a directed
graph, and a topological order of these components, as arcs are added.  This
extension also has an amortized time bound of $O(m^{1/2})$ per arc.

\end{abstract}

%% file: intro.tex
\section{Introduction} \label{sec:intro}

We consider three related problems on dynamic directed graphs: cycle detection;
maintaining a topological order; and maintaining strong components, along with a
topological order of them.  Cycle detection and maintaining a topological order are
closely connected, since a directed graph has a topological order if and only if it
is acyclic.  Previous work has focused mainly on the topological ordering problem;
the problem of maintaining strong components has received little attention.  We
present a solution to both problems that improves the best known time bound for
sparse graphs by a logarithmic factor.

A {\em topological order} $O$ of a directed graph is a total order of the vertices
such that for every arc $(v,w)$, $O(v) < O(w)$.  A directed graph has a topological
order (and in general more than one) if and only if it is acyclic \cite{Harary1965}.
A directed graph is {\em strongly connected} if, for each pair of vertices $v$ and
$w$, there is a path from $v$ to $w$.  The {\em strongly connected components} of a
directed graph are its maximal strongly connected induced subgraphs.  These
components partition the vertices. Given a directed graph, its {\em graph of strong
components} is the multigraph whose vertices are the components of the given graph,
with an arc $(c(v),c(w))$ for each arc $(v,w)$ of the original graph, where $c(x)$ is
the component containing vertex $x$. Ignoring loops, the graph of strong components
is acyclic; thus the components can be topologically ordered.

Given a fixed $n$-vertex, $m$-arc graph, a topological order can be found in $O(n +
m)$ time by either of two algorithms: repeated deletion of sources
\cite{Knuth1973,Knuth1974} or depth-first search \cite{Tarjan1972}.  The former
method extends to the enumeration of all possible topological orderings.  Strong
components, and a topological order of them, can also be found in $O(n + m)$ time by
depth-first search, either one-way \cite{Tarjan1972,Cheriyan1996,Gabow2000} or
two-way \cite{Aho1983,Sharir1981}.

In some applications, the graph is not fixed but changes over time.  The {\em
incremental topological ordering problem} is that of maintaining a topological order
of a directed graph as arcs are added, stopping when addition of an arc creates a
cycle.  This problem arises in incremental circuit evaluation \cite{Alpern1990},
pointer analysis \cite{Pearce2003}, management of compilation dependencies
\cite{Marchetti1993,Omohundro1992}, and deadlock detection \cite{Belik1990}.  In some
applications cycles are not fatal; strong components, and a topological order of
them, must be maintained.  For example, pointer analysis can optimize work based on
cyclic relationships \cite{Pearce2003b}.

In considering the incremental topological ordering and strong components problems,
we shall assume that the vertex set is fixed and specified initially, and that the
arc set is initially empty.  One can easily extend our ideas to support vertex as
well as arc additions, with a time of $O(1)$ per vertex addition.  The topological
ordering methods we consider also easily handle arc deletions, since deletion of an
arc cannot create a cycle nor invalidate the current topological order, but our time
bounds are only valid if there are no arc deletions.  Maintaining strong components
under arc deletions is a harder problem; we briefly discuss previous work on it in
Section \ref{sec:open}.

One could of course recompute a topological order, or strong components, after each
arc addition, but this would require $O(n + m)$ time per arc.  Our goal is to do
better.  The incremental topological ordering problem has received much attention,
especially recently.  Marchetti-Spaccamela et al. \cite{Marchetti1996} gave an
algorithm that takes $O(n)$ amortized time per arc addition.  Alpern et
al. \cite{Alpern1990} gave an algorithm that handles batched arc additions and has a
good time bound in an incremental model of computation.  Katriel and Bodlaender
\cite{Katriel2006} showed that a variant of the algorithm of Alpern et al. takes
$O(\min\{m^{1/2}\log n, m^{1/2} + (n^{2}\log n)/m\})$ amortized time per arc
addition. Liu and Chao \cite{Liu2007} tightened this analysis to $\Theta(m^{1/2} +
n^{1/2}\log n)$ per arc addition. Pearce and Kelly \cite{Pearce2006} gave an
algorithm with an inferior asymptotic time bound that they claimed was fast in
practice on sparse graphs.  Ajwani et al. \cite{Ajwani2006} gave an algorithm with an
amortized time per arc addition of $O(n^{2.75}/m)$, which for dense graphs is better
than the bound of Liu and Chao.  Finally, Kavitha and Mathew \cite{Kavitha2007}
improved the results of both Liu and Chao and Ajwani et al. by presenting another
variant of the algorithm of Alpern et al. with an amortized time bound per arc
addition of $O(m^{1/2} + (n\log n)/m^{1/2})$, and a different algorithm with an
amortized time bound per arc addition of $O(n^{2.5}/m)$.

The problem of maintaining strong components incrementally has received much less
attention.  F{\"a}hndrich et al. \cite{Fahndrich1998} gave an algorithm that searches
for the possible new strong component after each arc addition; their algorithm does
not maintain a topological order of components.  Pearce \cite{Pearce2005} and Pearce
et al. \cite{Pearce2003b} mention extending the topological ordering algorithm of
Pearce and Kelly and that of Marchetti-Spaccamela et al. to the strong components
problem, but they provide very few details.

We generalize the bidirectional search method introduced by Alpern et al. for
incremental topological ordering and used in later algorithms.  We identify the
properties of this method that allow older algorithms to achieve their respective
running times. We observe that the method does not require an ordered search (used in
all previous algorithms) to be correct.  This allows us to refine the general method
into one that avoids the use of heaps (priority queues), but instead uses either
approximate median-finding or random selection, resulting in an $O(m^{1/2})$ amortized
time bound per arc addition. The randomized version of our algorithm is especially
simple. We extend our incremental topological ordering algorithm to provide the first
detailed algorithm to maintain strong components and a topological order of them.
This algorithm also takes $O(m^{1/2})$ amortized time per arc addition.

The body of our paper is organized as follows.  Section \ref{sec:bi-search} describes
the bidirectional search method, verifies its correctness, and analyzes its running
time.  Section \ref{sec:c-search} refines the method to yield an algorithm fast for
sparse graphs.  Section \ref{sec:impl} provides some implementation details. Section
\ref{sec:strong} extends the algorithm to the problem of maintaining strong
components.  Section \ref{sec:open} examines lower bounds and other issues. In
particular, we show in this section that on sparse graphs our $O(m^{1/2})$ amortized
time bound is tight to within a constant factor for a natural class of algorithms
that includes all the existing ones.

%% file: bi-search.tex
\section{Topological Ordering via Bidirectional Search} \label{sec:bi-search}

We develop our topological ordering algorithm through refinement of a general method
that encompasses many of the older algorithms.  By vertex order we mean the current
topological order.  We maintain the vertex order using a data structure such that
testing the predicate ``$x < y$'' for two vertices $x$ and $y$ takes $O(1)$ time, as
does deleting a vertex from the order and reinserting it just before or just after
another vertex.  The dynamic ordered list structures of Dietz and Sleator
\cite{Dietz1987} and Bender et al. \cite{Bender2002} meet these requirements: their
structures take $O(1)$ time worst-case for an order query or a deletion, and $O(1)$
time for an insertion, amortized or worst-case depending on the structure.  For us an
amortized bound suffices. In addition to a topological order, we maintain the
outgoing and incoming arcs of each vertex.  This allows bidirectional search.
Initially the vertex order is arbitrary and all sets of outgoing and incoming arcs
are empty.

To add an arc $(v, w)$ to the graph, proceed as follows.  Add $(v, w)$ to the set of
arcs out of $v$ and to the set of arcs into $w$.  If $v > w$, do a bidirectional
search forward from $w$ and backward from $v$ until finding either a cycle or a set
of vertices whose reordering will restore topological order.

A vertex is {\em forward} if it is $w$ or has been reached from $w$ by a path of arcs
traversed forward, {\em backward} if it is $v$ or has been reached from $v$ by a path
of arcs traversed backward.  A vertex is {\em scanned} if it is forward and all its
outgoing arcs have been traversed, or it is backward and all its incoming arcs have
been traversed.  To do the search, traverse arcs forward from forward vertices and
backward from backward vertices until either a forward traversal reaches a backward
vertex $x$ or a backward traversal reaches a forward vertex $x$, in which case there
is a cycle, or until there is a vertex $s$ such that all forward vertices less than
$s$ and all backward vertices greater than $s$ are scanned.

In the former case, stop and report a cycle consisting of a path from $w$ to $x$
traversed forward, followed by a path from $x$ to $v$ traversed backward, followed by
$(v, w)$.  In the latter case, restore topological order as follows.  Let $X$ be the
set of forward vertices less than $s$ and $Y$ the set of backward vertices greater
than $s$.  Find topological orders of $O_X$ and $O_Y$ of the subgraphs induced by $X$
and $Y$, respectively.  Assume $s$ is not forward; the case of $s$ not backward is
symmetric.  Delete the vertices in $X \cup Y$ from the current vertex order and
reinsert them just after $s$, in order $O_Y$ followed by $O_X$.  (See Figure
\ref{fig:reorder}.)

\begin{figure}
\centering
\includegraphics[scale=0.5]{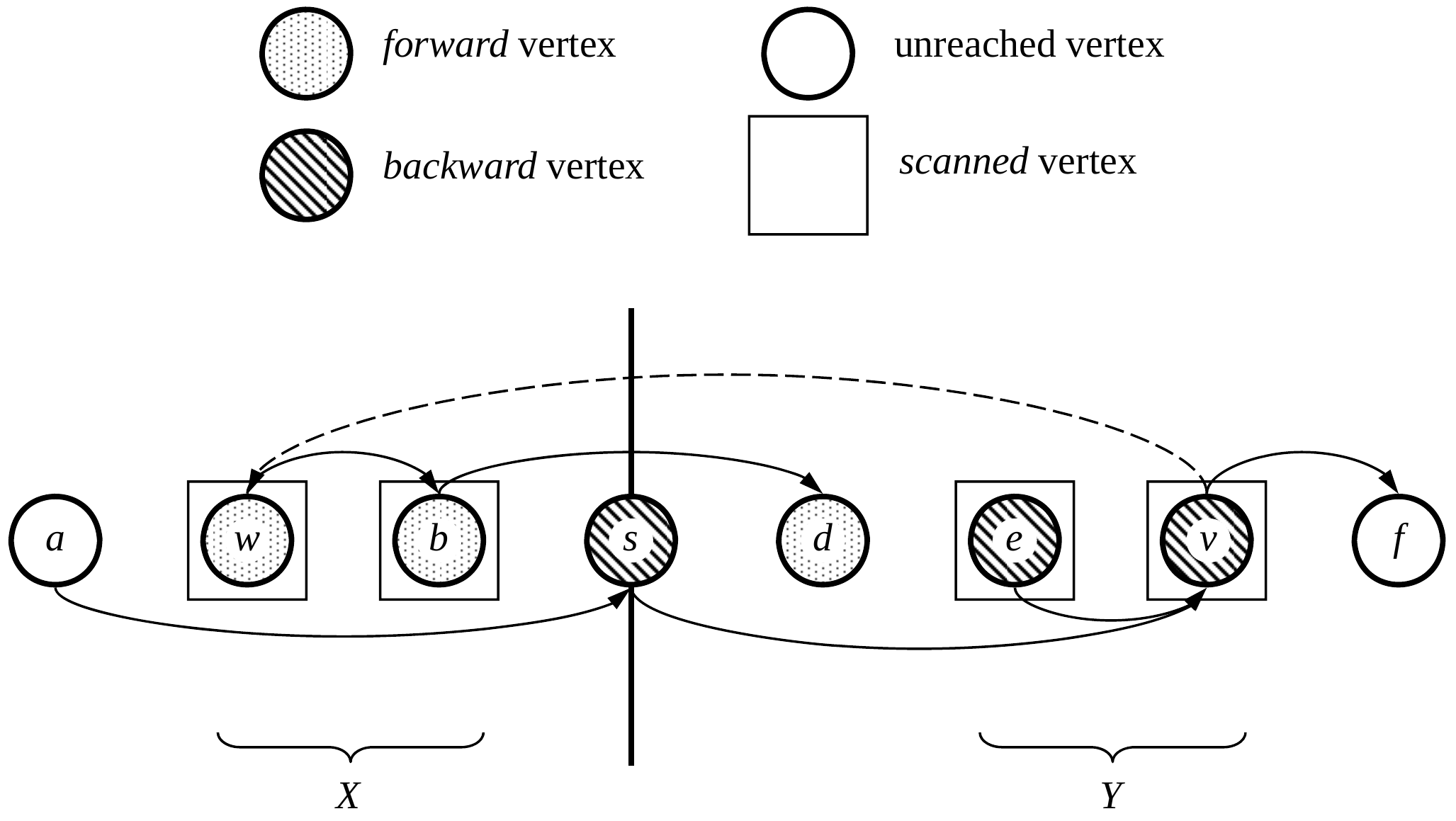} \\
\vspace{0.1in}
\includegraphics[scale=0.5]{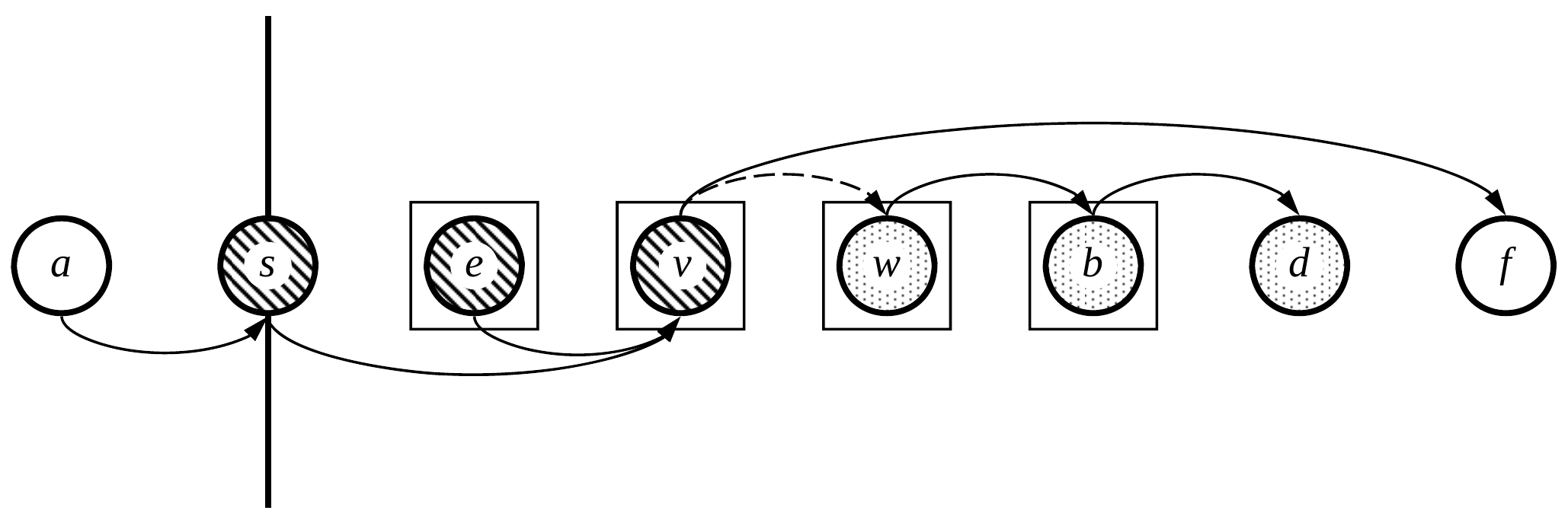} \\
\caption{Restoring topological order after bidirectional search. Since $s$ is not
forward, the vertices in $Y$ are inserted (in topological order) 
just after $s$, followed by the vertices in $X$.}
\label{fig:reorder}
\end{figure}

\begin{theorem} \label{thm:bi-search-cor}

The incremental topological ordering method is correct.

\end{theorem}

\begin{proof}

We prove by induction on the number of arc additions that the method maintains a
topological order until it stops and reports a cycle.  Consider the bidirectional
search triggered by the addition of an arc $(v, w)$ with $v > w$.  The search must
stop because it eventually runs out of arcs to traverse; when this happens, all
forward and backward vertices are scanned, which means that the second stopping
condition holds for any $s$ (although it may hold earlier for some particular
$s$). If a forward arc traversal during the search reaches a backward vertex $x$ or a
backward arc traversal reaches a forward vertex $x$, then there is a cycle,
consisting of a path from $w$ to $x$ traversed forward, followed by a path from $x$
to $v$ traversed backward, followed by $(v, w)$.

Conversely, suppose that the addition of $(v, w)$ creates a cycle, consisting of $(v,
w)$ followed by a path $P$ from $w$ to $v$.  Consider the situation when the search
stops.  Let $x$ be the first non-forward vertex on $P$ and $y$ the last non-backward
vertex on $P$; both $x$ and $y$ exist since $v$ is backward (non-forward) and $w$ is
forward (non-backward).  Let $u$ be the vertex preceding $x$ on $P$ and $z$ the
vertex following $y$ on $P$.  Then $u$ is forward and $z$ is backward.  Traversal of
the arc $(u, x)$ causes either a cycle to be reported (if $x$ is backward) or $x$ to
become forward (if it is previously unreached).  The latter contradicts the choice of
$x$.  Thus either the search stops and reports a cycle, or $u$ is unscanned.
Symmetrically, either the search stops and reports a cycle or $z$ is unscanned.
Since $u < z$, for any $s$ either $u < s$ or $z > s$.  Thus the search cannot stop
without reporting a cycle.  We conclude that the method reports a cycle if and only
if the addition of $(v, w)$ creates one.

Suppose the addition of $(v, w)$ does not create a cycle.  Then the search cannot
report a cycle.  Thus, for some $s$, the search will stop with all forward vertices
less than $s$ and all backward vertices greater than $s$ scanned.  We need to show
that the new vertex order is topological, assuming that the old one was.  Assume $s$
is not forward; the case of $s$ not backward is symmetric. In the new order $v < w$,
since $v \in X$ and $w \in Y \cup \{s\}$.  Let $(x, y)$ be any arc other than $(v,
w)$.  We do a case analysis based on which of the sets $X$ and $Y$ contain $x$ and
$y$.  There are nine cases, but two are symmetric and two are impossible, reducing
the number of cases to six. \\

\noindent {\bf Case 1.} $\{x, y\} \subseteq X$ or $\{x, y\} \subseteq Y$: $x < y$ in
the new order since $x$ and $y$ are both in $O_X$ or both in $O_Y$.

\noindent {\bf Case 2.} $x \in X$ and $y \notin X$: $y \notin Y$ since the graph is
acyclic.  Since $x$ is scanned, $y > s$ in the old order, which implies $x < y$ in
the new one.

\noindent {\bf Case 3.} $x \in Y$ and $y \notin Y$: $y \notin X$ since $(x, y) \neq
(v, w)$. Since $x \in Y$, $y > s$ in the old order, which implies $x < y$ in the new
one.

\noindent {\bf Case 4.} $x \notin X \cup Y$ and $y \in X$: $x < y$ and $y < s$ in the
old order, which implies $x < y$ in the new one.

\noindent {\bf Case 5.} $x \notin X \cup Y$ and $y \in Y$: $x \le s$ in the old order
since $y$ is scanned and $x \notin Y$.  In the new order, $x \le s < y$.

\noindent {\bf Case 6.} $x \notin X \cup Y$ and $y \notin X \cup Y$: $x < y$ in both
the old and new orders.  \qed

\end{proof}

\begin{lemma} \label{lem:bi-search-run}

The time per arc addition is $O(1)$ plus $O(1)$ per arc traversed by the search
plus any overhead needed to guide the search.

\end{lemma}

\begin{proof}

The time for the bidirectional search is $O(1)$ plus $O(1)$ per arc traversed plus
overhead.  The subgraphs induced by $X$ and $Y$ contain only traversed arcs.  The
time to topologically order them is linear since a static topological ordering
algorithm suffices.  The time to delete vertices in $X \cup Y$ from the old vertex
order and reinsert them is $O(1)$ per vertex in $X \cup Y$.  \qed

\end{proof}

To obtain a good time bound we need to minimize both the number of arcs traversed and
the search overhead.  In our discussion we shall assume that no cycle is created.
Only one arc addition, the last one, can create a cycle; the last search takes $O(m)$
time plus overhead.

We need a way to charge the search time against graph changes caused by arc
additions.  To measure such changes, we count pairs of related graph elements, either
vertex pairs, vertex-arc pairs, or arc-arc pairs: two such elements are {\em related}
if they are on a common path.  The number of related pairs is initially zero and
never decreases.  There are at most ${n \choose 2} < n^2/2$ vertex-vertex pairs, $nm$
vertex-arc pairs, and ${m \choose 2} < m^2/2$ arc-arc pairs.  Of most use to us are the
related arc-arc pairs.

We limit the search in three ways to make it more efficient.  First, we restrict it
to the {\em affected region}, the set of vertices between $w$ and $v$.  Specifically,
only arcs $(u, x)$ with $u < v$ are traversed forward, only arcs $(y, z)$ with $z >
w$ are traversed backward.  This suffices to attain an $O(n)$ amortized time bound
per arc addition.  The bound comes from a count of newly-related vertex-arc pairs:
each arc $(u, x)$ traversed forward is newly related to $v$, each arc $(y, z)$
traversed backward is newly related to $w$.  The algorithm of Marchetti-Spaccamela et
al. \cite{Marchetti1996} is the special case that does just a unidirectional search
forward from $w$ using $s = v$, with one refinement and one difference: it does a
depth-first search, and it maintains the topological order as an explicit mapping
between the vertices and the integers from $1$ to $n$.

Unidirectional search allows a more space-efficient graph representation, since we
need only forward incidence sets, not backward ones.  But bidirectional search has a
better time bound if it is suitably limited.  We make the search {\em balanced}: each
traversal step is of two arcs concurrently, one forward and one backward.  There are
other balancing strategies \cite{Alpern1990,Katriel2006,Kavitha2007}, but this simple
one is best for us.  Balancing by itself does not improve the time bound; we need a
third restriction.  We call an arc $(u, x)$ traversed forward and an arc $(y, z)$
traversed backward {\em compatible} if $u < z$.  Compatibility implies that $(u, x)$
and $(y, z)$ are newly related.  We make the search {\em compatible}: each traversal
step is of two compatible arcs.

\begin{lemma} \label{lem:c-search-arcs}

If the searches are compatible, the amortized number of arcs traversed during
searches is $O(m^{1/2})$ per arc addition.

\end{lemma}

\begin{proof} 

We count related arc-arc pairs.  Consider a compatible search of $2k$ arc traversals,
$k$ forward and $k$ backward.  Order the arcs $(u, x)$ traversed forward in
increasing order on $u$, breaking ties arbitrarily.  Let $(u, x)$ be the $\lceil k/2
\rceil^{th}$ arc in the order.  Arc $(u, x)$ and each arc following $(u, x)$ has a
compatible arc $(y, z)$ traversed backward.  Compatibility and the ordering of
forward traversed arcs imply that $u < z$.  Thus each such arc $(y, z)$ is newly
related to $(u, x)$ and to each arc preceding $(u, x)$, for a total of at least
$(k/2)^2$ newly related pairs.

We divide searches into two kinds: those that traverse at most $m^{1/2}$ arcs and
those that traverse more.  Searches of the first kind satisfy the bound of the lemma.
Let $2k_i$ be the number of arcs traversed during the $i^{th}$ search of the second
kind. Since $2k_i > m^{1/2}$ and $\sum_i (k_i/2)^2 < m^2/2$, $\sum_i k_i < 2\sum_i
{k_i}^2 / m^{1/2} = 8 \sum_i (k_i/2)^2 / m^{1/2} < 4m^{3/2}$. Thus there are
$O(m^{1/2})$ arc traversals per arc addition. \qed

\end{proof}

We still need a way to do a compatible search.  The most straightforward way is to
make the search ordered: traverse arcs $(u, x)$ forward in non-decreasing order on
$u$ and arcs $(y, z)$ backward in non-increasing order on $z$.  We can implement an
ordered search using two heaps (priority queues) to store unscanned forward and
unscanned backward vertices.  In essence this is the algorithm of Alpern et
al. \cite{Alpern1990}, although they use a different balancing strategy.  The heap
overhead is $O(\log n)$ per arc traversal, resulting in an amortized time bound of
$O(m^{1/2} \log n)$ per arc addition.  More-complicated balancing strategies lead to
the improvements \cite{Katriel2006,Liu2007,Kavitha2007} in this bound for non-dense
graphs mentioned in Section \ref{sec:intro}.

%% file: c-search.tex
\section{Compatible Search via a Soft Threshold} \label{sec:c-search}

The running time of an ordered search can be reduced further, even for sparse graphs,
by using a faster heap implementation, such as those of van Emde Boas
\cite{Boas1977,Boas1977b}, Thorup \cite{Thorup2004}, and Han and Thorup
\cite{Han2002}.  But we can do even better, avoiding the use of heaps entirely, by
exploiting the flexibility of compatible search.  What we need is a way to find a
pair of candidate vertices for a traversal step: an unscanned forward vertex $u$ and
an unscanned backward vertex $z$ such that $u < z$.  It is easy to keep track of
unscanned forward and backward vertices, but if we have an unscanned forward vertex
$u$ and an unscanned backward vertex $z$, what do we do if $u > z$?  One answer is to
(temporarily) bypass one of them, but which one?  To resolve this dilemma, we use a
soft threshold $s$ that is an estimate of the stopping threshold for the search.  If
$u > z$, at least one is on the wrong side of $s$: either $u \ge s$ or $z \le s$ (or
both).  We call such a vertex {\em far}.  Our decision rule is to bypass a far
vertex.

In addition to $s$, we need two hard thresholds to bound candidate vertices, a low
threshold $l$ and a high threshold $h$.  We also need a partition of the candidate
forward vertices into two sets, $A$ and $B$, and a partition of the candidate
backward vertices into two sets, $C$ and $D$.  The vertices in $B \cup D$ are far; we
call the vertices in $A \cup C$ {\em near} (they may or may not be far).  To do a
compatible search, initialize $l$ to $w$, $h$ to $v$, $s$ to $v$ or $w$ or any vertex
in between, $A$ to $\{w\}$, $C$ to $\{v\}$, and $B$, $D$, $X$, and $Y$ to empty.
Then repeat an applicable one of the following cases until $A \cup B$ or $C \cup D$
is empty (see Figure \ref{fig:c-search}): \\

\noindent {\bf Case 1f.} $A$ is empty: move all vertices in $B$ to $A$, set $l = s$,
let $s$ be a vertex in $A$, and set $s = \min\{s, h\}$.

\noindent {\bf Case 1b.} $C$ is empty: move all vertices in $D$ to $C$, set $h = s$,
let $s$ be a vertex in $C$, and set $s = \max\{s, l\}$. \\

\noindent In the remaining cases $A$ and $C$ are non-empty.  Choose a vertex $u$ in
$A$ and a vertex $z$ in $C$. \\

\noindent {\bf Case 2f.} $u \ge h$: delete $u$ from $A$.

\noindent {\bf Case 2b.} $z \le l$: delete $z$ from $C$.

\noindent {\bf Case 3f.} $\max\{z, s\} \le u < h$: move $u$ from $A$ to $B$.

\noindent {\bf Case 3b.} $l < z \le \min\{u, s\}$: move $z$ from $C$ to $D$. \\

\noindent In the remaining cases $l \le u < z \le h$. \\

\noindent {\bf Case 4f.} All arcs out of $u$ are traversed: move $u$ from $A$ to $X$.

\noindent {\bf Case 4b.} All arcs into $z$ are traversed: move $z$ from $C$ to $Y$.

\noindent {\bf Case 5.} No other case applies: traverse an untraversed arc $(u, x)$ out of $u$
and an untraversed arc $(y, z)$ into $z$.  If $x$ is backward or $y$ is forward, stop
and report a cycle.  If $x$ is unreached, make it forward and add it to $A$; if $y$
is unreached, make it backward and add it to $C$. \\

\begin{figure}
\centering
\includegraphics[scale=0.5]{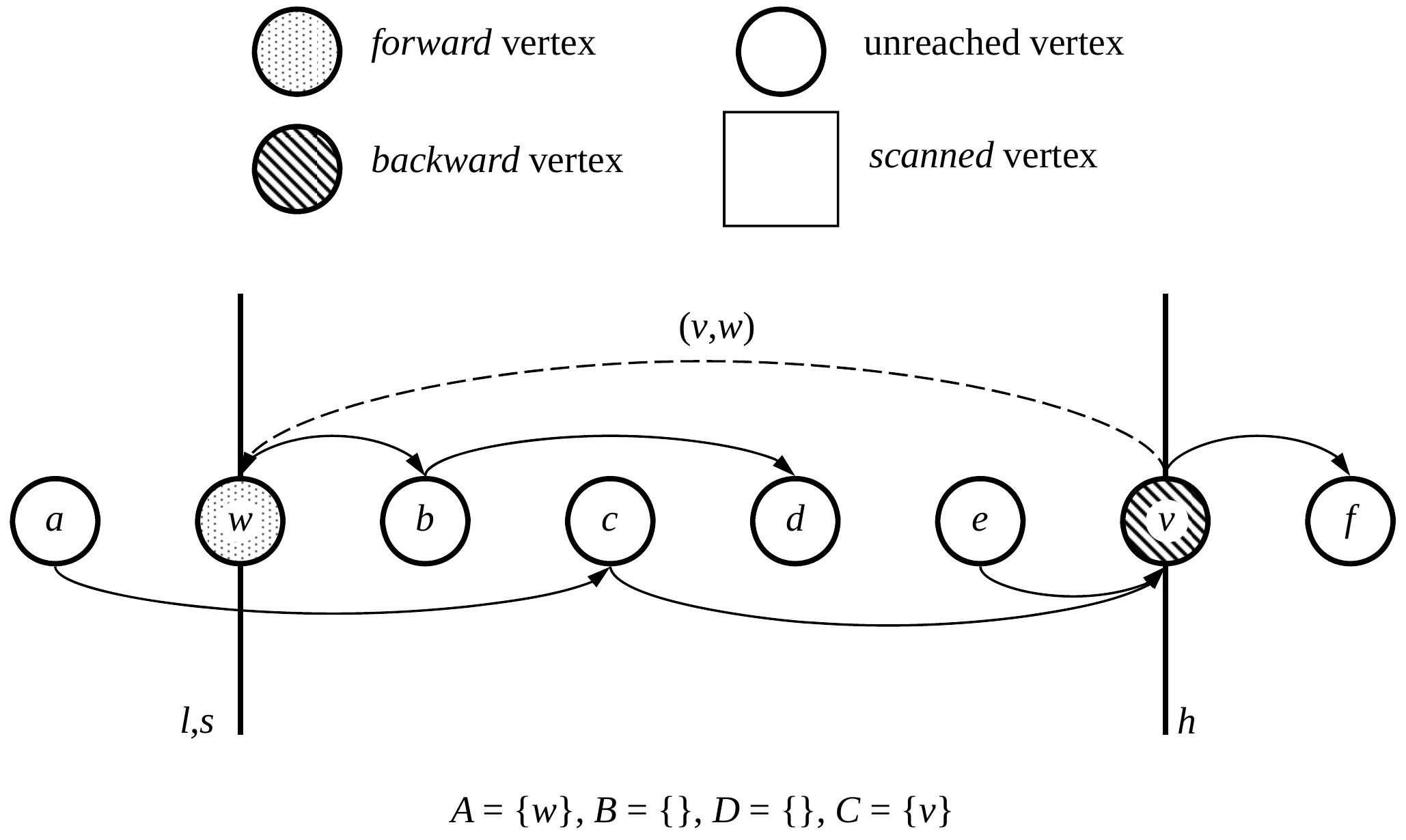} \\
\vspace{0.1in}
\includegraphics[scale=0.5]{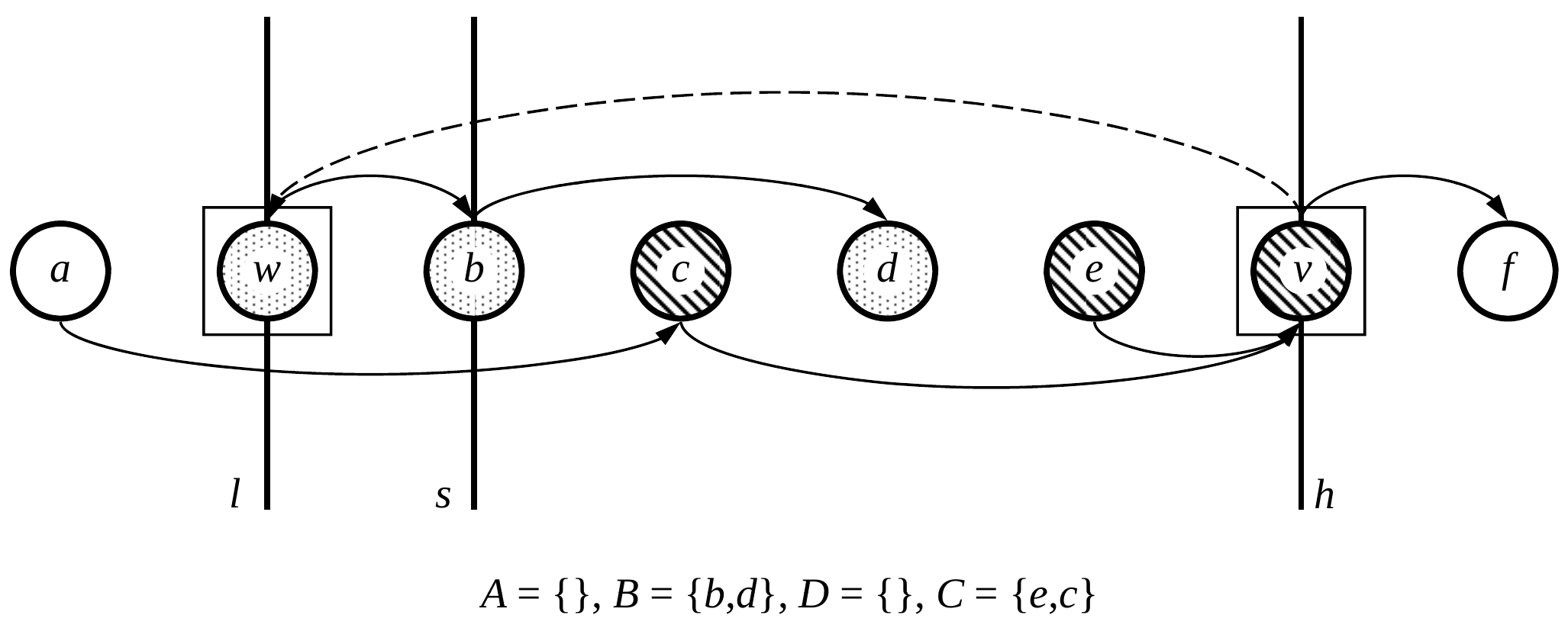} \\
\vspace{0.1in}
\includegraphics[scale=0.5]{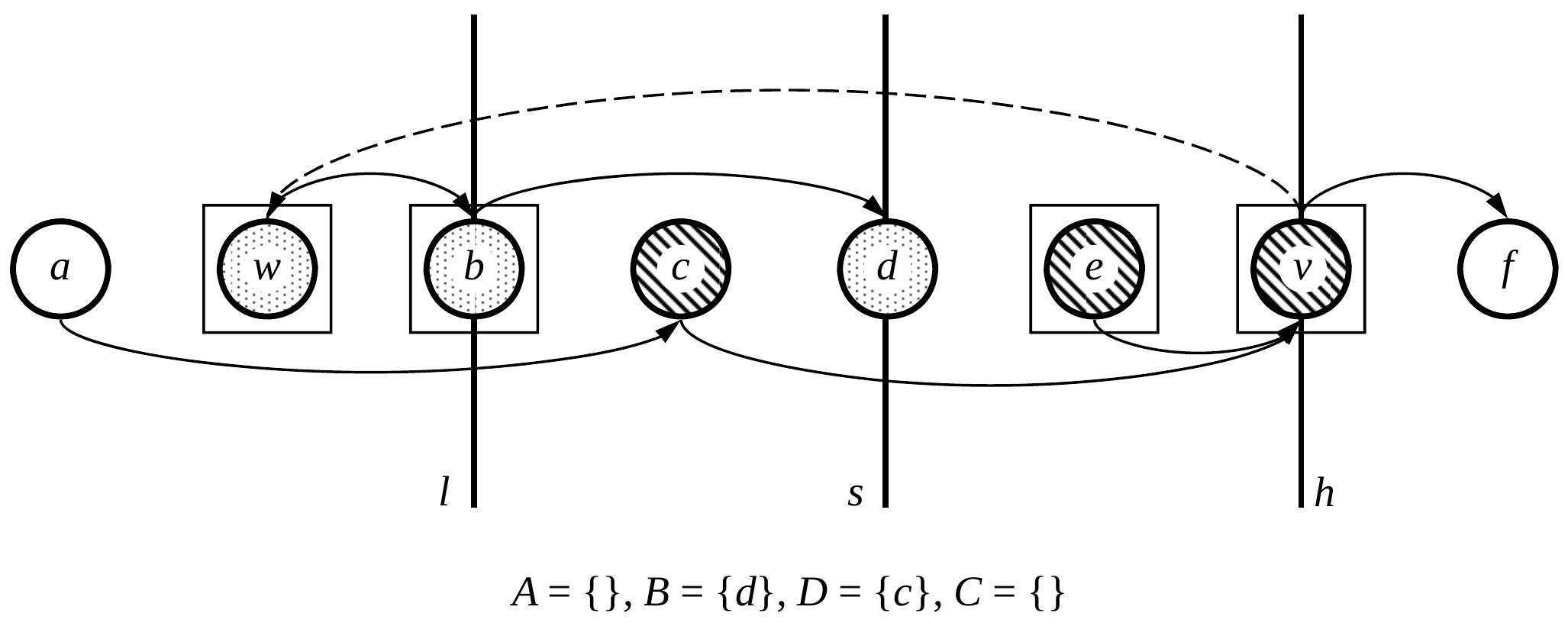} \\
\vspace{0.1in}
\includegraphics[scale=0.5]{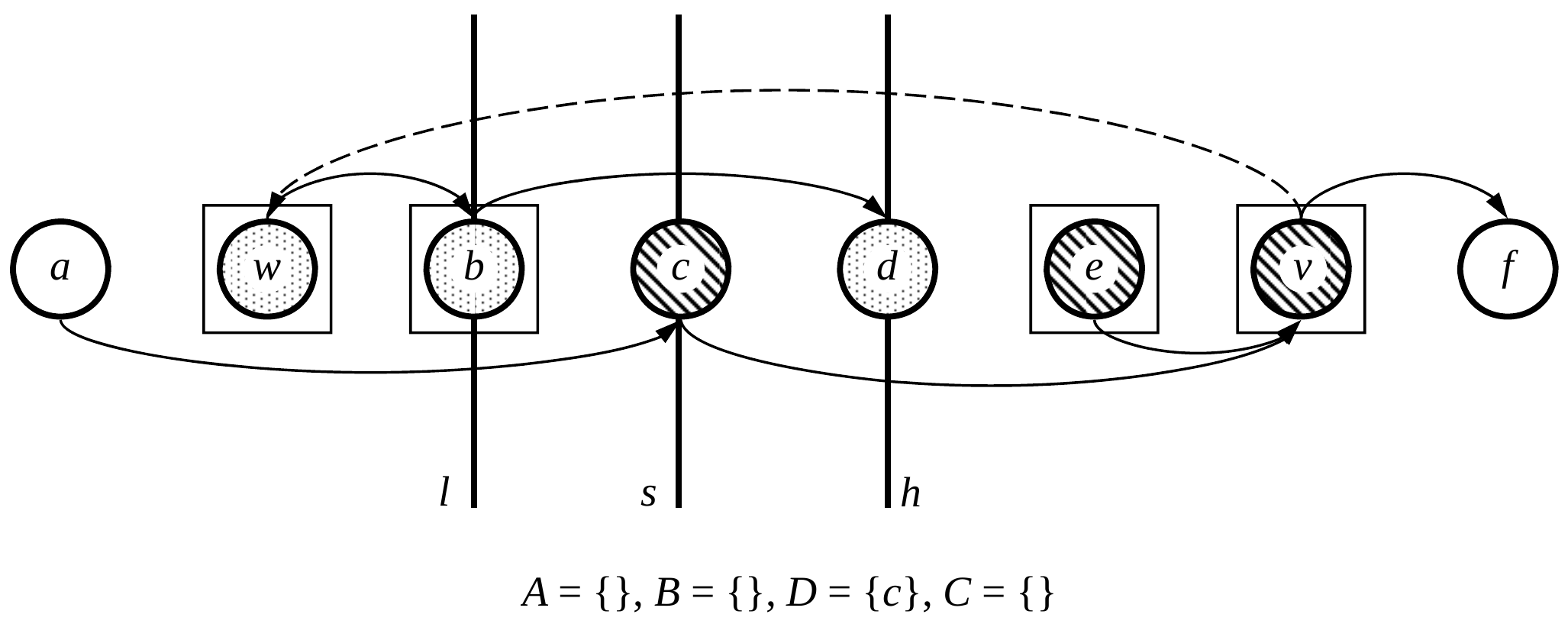} \\
\caption{Compatible search via a soft threshold. At the end of the search, $A \cup B$
is empty, so $s = h = d$ is the final threshold and the vertices are reordered as
described in Section \ref{sec:bi-search}.}
\label{fig:c-search}
\end{figure}

If the search stops without detecting a cycle, set $s = h$ if $A \cup B$ is
empty, $s = l$ otherwise.  Delete from $X$ all forward vertices no less than $s$ and from $Y$
all backward vertices no greater than $s$.  Reorder the vertices as described in
Section \ref{sec:bi-search}.

This algorithm does some deletions and bypasses lazily that could be done more
eagerly.  In particular, when a vertex is moved from the far set to the near set in
Case 1f or 1b, it can be deleted if it is on the wrong side of the corresponding hard
threshold (by applying Case 2f or 2b immediately), and it can be left in the far set
if it is on the wrong side of the new soft threshold.  The algorithm does traversal
steps eagerly: it can do such a step even if one of the two vertices involved is far,
since all that Case 5 requires is compatibility.  The results below apply to the
original algorithm and to variants that do deletions and bypasses more eagerly and
traversal steps more lazily.

\begin{theorem} \label{thm:c-search-cor}

Compatible search with a soft threshold is correct.

\end{theorem}

\begin{proof}

Compatible search is just a special case of the general method presented in Section
\ref{sec:bi-search}.  If the search stops with $A \cup B$ empty, then all
forward vertices less than $h$ are scanned; hence the stopping condition holds for $s =
h$.  The case of $C \cup D$ empty is symmetric.  \qed

\end{proof}

\begin{lemma}

The running time of compatible search with a soft threshold is $O(1)$ plus $O(1)$ per
arc traversed plus $O(1)$ for each time a vertex becomes near.

\end{lemma}

\begin{proof}

Each case either traverses two arcs and adds at most two vertices to $A \cup C$,
or permanently deletes a vertex from $A \cup C$, or moves a vertex from $A \cup
C$ to $B \cup D$, or moves one or more vertices from $B \cup D$ to $A
\cup C$.  The number of times vertices are moved from $A \cup C$ to $B \cup
D$ is at most the number of times vertices become near. \qed

\end{proof}

The algorithm is correct for any choice of soft threshold, but only a careful choice
makes it efficient.  Repeated cycling of vertices between near and far is the
remaining inefficiency.  We choose the soft threshold to limit such cycling no matter
how the search proceeds.  A good deterministic choice is to let the soft threshold be
the median or an approximate median of the appropriate set ($A$ or $C$); an
$\epsilon$-{\em approximate median} of a totally ordered set of $g$ elements is any
element that is no less than $\epsilon g$ and no greater than $\epsilon g$ of the
elements, for some constant $\epsilon > 0$. The median is a $1/2$-approximate
median. Finding the median or an approximate median takes $O(g)$ time
\cite{Blum1973,Schonhage1976}.  An alternative is to choose the soft threshold
uniformly at random from the appropriate set.  This gives a very simple yet efficient
randomized algorithm.

\begin{lemma} \label{lem:soft-near}

If each soft threshold is an $\epsilon$-approximate median of the set from which it
is chosen, then the number of times a vertex becomes near is $O(1)$ plus $O(1)$ per
arc traversed.  If each soft threshold is chosen uniformly at random, then the
expected number of times a vertex becomes near is $O(1)$ plus $O(1)$ per arc
traversed.

\end{lemma}

\begin{proof}

The value of $l$ never decreases as the algorithm proceeds; the value of $h$ never
increases.  Let $k$ be the number of arcs traversed.  Suppose each soft threshold is
an $\epsilon$-approximate median.  The first time a vertex is reached, it becomes
near.  Each subsequent time it becomes near, it is one of a set of $g$ vertices that
become near, as a result of being moved from $B$ to $A$ or from $D$ to $C$.  The two
cases are symmetric; consider the former.  We shall show that no matter what happens
subsequently, at least $\epsilon g$ vertices have become near for the last time.
Just after $s$ is changed, at least $\epsilon g$ vertices in $A$ are no less than
$s$, and at least $\epsilon g$ vertices in $A$ are no greater than $s$.  Just before
the next time $s$ changes, $l \ge s$ or $h \le s$.  In the former case, all vertices
no greater than $s$ can never again become near; in the latter case, all vertices no
less than $s$ can never again become near.  We charge the group of $g$ newly near
vertices to the vertices that become near for the last time.  The total number of
times vertices can become near is at most $(1 + 1/\epsilon)(2 + k)$, where $k$ is the
number of arcs traversed: there are at most $2 + k$ forward and backward vertices and
at most $1 + 1/\epsilon$ times a vertex can become near per forward or backward
vertex.

Essentially the same argument applies if the soft threshold is chosen uniformly at
random.  If a set of $g$ vertices becomes near, the expected number that become near
for the last time is at least $\sum_{1 \le i \le g/2} (2i)/g = (g/2 + 1)/2 > g/4$ if
$g$ is even, at least $\lceil g/2 \rceil + \sum_{1 \le i < \lfloor g/2 \rfloor}
(2i)/g = {\lceil g/2 \rceil}^2/g > g/4$ if $g$ is odd.  The total expected number of
times vertices can become near is at most $(5/4)(2 + k)$. \qed

\end{proof}

\begin{theorem} \label{thm:c-search-run}

The amortized time for incremental topological ordering via compatible search is
$O(m^{1/2})$ per arc addition, worst-case if each soft threshold is an
$\epsilon$-approximate median of the set from which it is chosen, expected if each
soft threshold is chosen uniformly at random.

\end{theorem}

\begin{proof}

Immediate from Lemmas \ref{lem:bi-search-run}-\ref{lem:soft-near}. \qed

\end{proof}

%% file: impl.tex
\section{Implementation} \label{sec:impl}

In this section we fill in some implementation details.  We also give alternative
implementations of dynamic ordered lists and the reordering step.

We need a way to keep track of traversed and untraversed arcs.  We maintain each
incidence set as a singly-linked list, with a pointer indicating the first
untraversed arc.  Each time a vertex is reached during a search, the pointer for its
appropriate incidence list is reset to the front.  Since a vertex cannot be both
forward and backward during a single search, one pointer per vertex suffices.

We also need a way to report cycles that are detected, if this is required of the
application.  For this purpose we maintain, for each forward and backward vertex
other than $v$ and $w$, the arc by which it is reached.

The description of compatible search is quite general; in particular, it does not
specify how to maintain the sets of candidate vertices, or the order in which to
consider candidates.  A simple implementation is to store the candidate forward
vertices in a deque \cite{Knuth1973} (double-ended queue) $F$, with vertices in $A$
preceding those in $B$ and a pointer to the first vertex in $B$.  The first vertex on
$F$ is $u$. New forward vertices are added to the front of $F$; bypassed far vertices
are moved from the front to the back of $F$.  When the pointer to the first vertex in
$B$ indicates the first vertex on $F$, this pointer is moved to the back, and $l$ and
$s$ are updated.  A deque $R$ that operates in the same way stores the backward
vertices.  Both $F$ and $R$ are actually steques \cite{Kaplan1995} (stack-ended
queues or output-restricted deques \cite{Knuth1973}), since deletions are only from
the front.  Thus each can be implemented as an array or as a singly-linked, possibly
circular list.

The dynamic ordered list implementations of Dietz and Sleator \cite{Dietz1987} and of
Bender et al. \cite{Bender2002} are two-level structures that store the elements, in
our case vertices, in contiguous blocks of up to $b$ elements, with $b = O(\log n)$,
and store the blocks in a doubly-linked list.  The blocks have numbers that are
consistent with the list order, higher numbers for later blocks.  The elements within
a block have numbers that are consistent with the order within the block.  Thus each
element has a two-part number that can be retrieved in $O(1)$ time and used to test
order in $O(1)$ time, worst-case.  The two methods differ in the details of
renumbering when insertions take place.  Deletions require no explicit renumbering
and take $O(1)$ time worst-case. Insertion takes $O(1)$ amortized time.  This bound
becomes worst-case if incremental updating is done.

For us an amortized bound suffices.  Even though the renumbering schemes of Dietz and
Sleator and of Bender et al. are not too complicated, we can use a simpler two-level
structure if we are willing to suffer an $O(n^{1/2})$ additive overhead per arc
addition, which does not affect the asymptotic bound.  We use a block size of
$O(n^{1/2})$, we completely renumber the elements within a block when the block
contents change other than by a deletion, and we completely renumber the blocks when
a block is inserted or deleted.  We number the blocks, and the elements within a new
block, with consecutive integers starting from 1.  When a block becomes less than
half full as a result of deletions, we combine it with one of its neighbors.  With
this method the time for a set of $k$ consecutive insertions is $O(n^{1/2} + k)$, and
the amortized time per deletion is $O(1)$.  This data structure is essentially a
two-level B-tree.  Using more levels with a smaller bound on block size decreases the
additive overhead per arc addition but increases the constant factor for order
queries.

We can topologically order the subgraphs induced by $X$ and $Y$ by using either of
the two linear-time algorithms for static topological ordering mentioned in Section
\ref{sec:intro}.  The subgraph induced by $X$ contains exactly the vertices less than
$s$ that are reachable from $w$.  Thus a simple method to order $X$ is to do a
depth-first search forward from $w$, traversing arcs only from vertices less than
$s$, and order the reached vertices less than $s$ in decreasing postorder
\cite{Tarjan1972}.  A symmetric depth-first search backward from $v$ orders $Y$.  An
alternative to using a static topological ordering method is to sort $X$ and $Y$ on
the current vertex order.  If we use binary comparisons, the time to sort is $O(k\log
k)$, where $k$ is the number of vertices sorted.  This method has the disadvantage
that it does not preserve the $O(m^{1/2})$ amortized bound per arc addition but
incurs a logarithmic overhead.  Alternatively we can take advantage of the vertex
numbering used by the dynamic ordered list structure.  Since the number of bits
needed to represent the vertex numbers is $O(\log n)$, we can sort in $O(n^{1/2} +
k)$ time via a radix sort with a fixed number of passes; the additive term can be
decreased at the expense of increasing the constant factor.  Radix sorting preserves
the $O(m^{1/2})$ bound per arc addition.

Which implementation of the algorithm is fastest in practice is a question to be
resolved by experiments.

%% file: strong.tex
\section{Maintenance of Strong Components} \label{sec:strong}

A natural extension of our topological ordering algorithm is to the problem of
maintaining strong components, and a topological order of them, as arcs are added.
This problem has received much less attention than the topological ordering problem.
Pearce \cite{Pearce2005} and Pearce et al. \cite{Pearce2003b} mention extending their
topological ordering algorithm and that of Marchetti-Spaccamela et
al. \cite{Marchetti1996} to the strong components problem, but they provide very few
details.  We shall describe how to extend the general method of Section
\ref{sec:bi-search} and the specific method of Section \ref{sec:c-search} to the
strong components problem, while maintaining for the latter the $O(m^{1/2})$
amortized time bound per arc addition.  We describe only the additions and changes
needed.  We begin at a high level and then fill in the implementation details.

The graph of strong components contains a vertex for each strong component, and, for
each original arc $(v, w)$, an arc $(c(v), c(w))$, where $c(x)$ is the vertex for the
strong component containing original vertex $x$.  The graph of strong components can
contain loops (arcs of the form $(x, x)$) and multiple arcs; we allow both, but
ignore loops when searching.  To maintain a topological order of the graph of strong
components (ignoring loops), we extend the method of Section \ref{sec:bi-search} as
follows:

To add an arc $(v, w)$, begin by finding $c(v)$ and $c(w)$.  Add $(c(v), c(w))$ to
the set of arcs out of $c(v)$ and to the set of arcs into $c(w)$.  If $c(v) > c(w)$,
do a bidirectional search forward from $c(w)$ and backward from $c(v)$, but ignore
loops and do not stop when a cycle is detected but allow vertices to be both forward
and backward.  Specifically, do the search as follows.  When about to traverse a
loop, put it aside instead of traversing it, so it will never be examined in future
searches.  When reaching a backward vertex $x$ by a forward traversal, do not stop,
but instead make $x$ forward. When reaching a forward vertex $x$ by a backward
traversal, do not stop, but instead make $x$ backward.  Call a vertex {\em forward
scanned} if all its outgoing arcs have been traversed, {\em backward scanned} if all
its incoming arcs have been traversed.  Continue the search until there is a vertex
$s$ such that all forward vertices less than $s$ are forward scanned and all backward
vertices greater than $s$ are backward scanned.

The addition of $(v, w)$ creates a new cycle, and a single new strong component, if
and only if at the end of the search some vertex is both forward and backward.  If
there is no such vertex, reorder the vertices as in Section \ref{sec:bi-search}.
Otherwise, let $X$ be the set of forward vertices less than $s$ and $Y$ the set of
backward vertices greater than $s$.  The vertex set of the new component is a subset
of $X \cup Y \cup \{s\}$.  Find the vertex set $Z$ of the new component and contract
the vertices in $Z$ to form a single new vertex $x$ representing the component.  Find
topological orders $O_X$ and $O_Y$ of the vertices in $X - Z$ and the vertices in $Y
- Z$.  Delete all vertices in $X$ and $Y$ from the vertex order.  If $s$ is in the
component, replace $s$ by $x$ in the order. (Insert $x$ after $s$ and delete $s$.)
Otherwise, insert $x$ just after $s$ if $s$ is not forward, just before $s$
otherwise.  Reinsert the vertices in $X - Z$ just after $x$ in order $O_X$ and insert
the vertices in $Y - Z$ just before $x$ in order $O_Y$.

A simple extension of the proof of Theorem \ref{thm:bi-search-cor} shows that this
algorithm is correct.  An arc can only become a loop once and be put aside once.  If
a new component is created, the time to find it is linear in the number of arcs
traversed by the search if one of the linear-time static algorithms for finding
strong components mentioned in Section \ref{sec:intro} is used.  (We describe a
simpler alternative below.)

We can use compatible search with a soft threshold in this algorithm, the only change
being to allow vertices to be both forward and backward.  The proof of Theorem
\ref{thm:c-search-cor} gives correctness.

Before analyzing the running time of the algorithm, we fill in some implementation
details.  We represent the vertex sets of the strong components using a disjoint set
data structure \cite{Tarjan1975,Tarjan1984}.  Each set has a distinguished vertex,
defined by the data structure, that represents the set.  Two operations are possible:

\begin{itemize}

\item $find(x)$: return the representative of the set containing $x$;

\item $link(x, y)$: given representatives $x$ and $y$ of two different sets, form their
union, and choose one of $x$ and $y$ to represent it.

\end{itemize}

Originally each vertex is in a singleton set.

We maintain the arcs in their original form, with their original end vertices.  For
each representative of a component, we maintain the set of arcs out of component
vertices and the set of arcs into component vertices.  These sets are singly linked
circular lists, so that catenation takes $O(1)$ time.  Each representative also has
two pointers, to the first untraversed forward arc and the first untraversed backward
arc.  To process an arc $(v, w)$ during a search, put $(v,w)$ aside (delete it from
the incidence list) if $find(v) = find(w)$; otherwise, traverse $(find(v),
find(w))$. To form a new strong component, combine representatives into a single
component by using $link$ to combine the corresponding vertex sets and catenating all
the forward incidence lists and all the backward incidence lists of the
representatives to form the incidence lists of the representative of the new
component.

With this data structure the time per $link$ is $O(1)$ \cite{Tarjan1975,Tarjan1984}.
Each catenation of arc lists takes $O(1)$ time per $link$.  The total time for all
the $link$s and all the catenations of arc lists over all arc additions is
$O(\min\{n, m\})$, which is $O(1)$ per arc addition.  There are two $find$s per arc
examined during the search. Such an arc is either found to be a loop and never
examined again, or it is traversed by the search. Lemmas
\ref{lem:bi-search-run}-\ref{lem:soft-near} and Theorem \ref{thm:c-search-run} hold
for the extension to strong components, not counting the time for $find$s.  The
number of $find$s is $O(m^{1/2})$ per arc addition by Lemma \ref{lem:c-search-arcs}.
Since the bound on the ratio of $find$s to $link$s is so high ($O(m^{1/2})$), the
amortized time for $find$s is $O(m^{1/2})$ per arc addition, even if the disjoint set
implementation uses path compression with na{\"i}ve linking \cite{Tarjan1984}; this
bound also holds for path compression with linking by rank or size
\cite{Tarjan1984}. The total amortized time per arc addition is $O(m^{1/2})$.

We conclude this section with three remarks about implementation.  First, there is a
simple way to topologically order $X$ and $Y$ and to find the new strong component
when one is created.  The method extends the topological ordering method mentioned in
Section \ref{sec:impl}.  Do a forward depth-first search from $w$, traversing arcs
only from vertices less than $s$.  When traversing an arc $(x, y)$, if $y$ is
backward, make $x$ backward.  When the search stops, the vertices in $X \cap Z$ are
the vertices reached by the search that are backward; decreasing postorder is a
topological order of $X - Z$.  A symmetric backward search from $v$ gives the
vertices in $Y \cap Z$ and topologically orders $Y - Z$.  Vertex $s$ is in the
component if it is both forward and backward.

Second, our representation of strong components does not provide a way to list the
vertices in each component.  To allow this, maintain for each component a circularly
linked list of the vertices in it: catenating such lists when a new component is
formed takes $O(1)$ time per $link$, and the vertices of a component can be listed in
$O(1)$ time per vertex by starting at the representative (or at any vertex) and
traversing the list.

Third, our method does not maintain a representation of the original graph. To
represent both the graph of strong components and the original graph, maintain, for
each vertex, four lists of its incident arcs: outgoing but not yet found to be loops;
outgoing loops; incoming but not yet found to be loops; incoming loops.  The first
and third lists are also part of the lists of outgoing and incoming arcs of the
component representative; when the end of a sublist is reached during a search, this
can be detected in $O(1)$ time by checking one end of the next arc.  Each arc is on
up to four lists, two of outgoing arcs and two of incoming arcs, but it only needs
two pointers to arcs next on a list.  (See Figure \ref{fig:strong}.)

\begin{figure}
\centering
\includegraphics[scale=0.5]{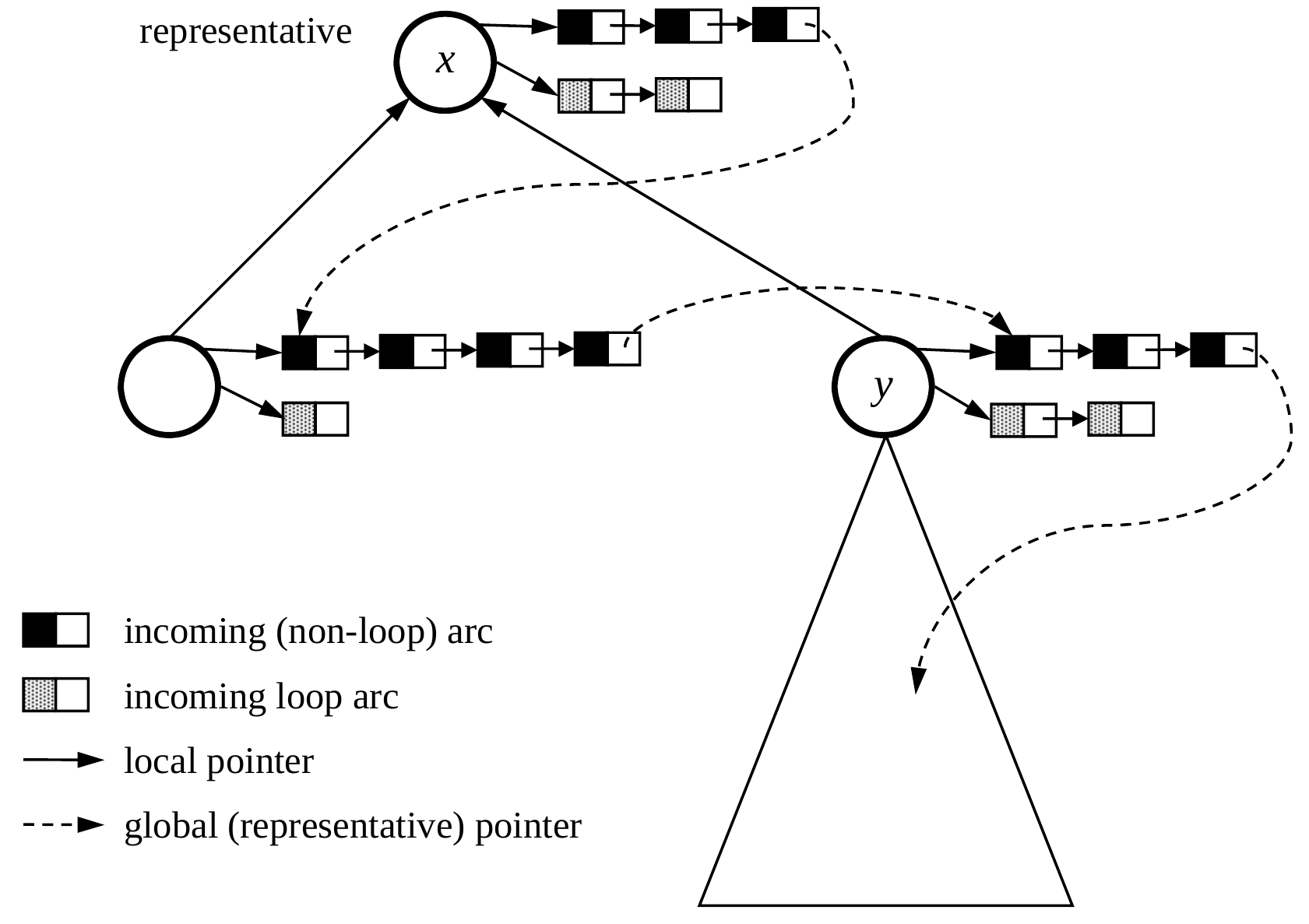} \\
\caption{Data structure for maintaining the graph of strong components and the
original graph. Only incoming arc lists are shown. Component $x$ and $y$ were
recently linked; all non-loop arcs in the new component are accessible from $x$.}
\label{fig:strong}
\end{figure}

%% file: open.tex
\section{Lower Bounds and Other Issues} \label{sec:open}

We have described an incremental topological ordering algorithm that takes
$O(m^{1/2})$ amortized time per arc addition. In combination with the $O(n^{2.5}/m)$
bound of the second algorithm of Kavitha and Mathew, this gives an overall bound of
$O(\min\{m^{1/2}, n^{2.5}/m\})$ per arc addition. For $m/n = o(\log n)$, our bound is
asymptotically smaller than that of Kavitha and Mathew ($O(m^{1/2} + (n\log
n)/m^{1/2})$), by a logarithmic factor for $m/n = O(1)$. Although our approach is
based on that of Alpern et al. \cite{Alpern1990} and its variants, our method avoids
the use of heaps.  The randomized version of our algorithm is especially simple. Like
previous algorithms, our algorithm can handle arc deletions, since a deletion
preserves the validity of the current topological order, but the time bound is no
longer valid.  Katriel and Bodlaender provide a class of examples on which our
algorithm runs in $\Omega(m^{1/2})$ time per arc addition; thus our analysis is
tight.  We have extended the algorithm to the problem of maintaining strong
components and a topological order of strong components.  In this problem, arc
deletions are harder to handle, since one arc deletion can cause a strong component
to split into several smaller ones. Roditty and Zwick \cite{Roditty2002} presented a
randomized algorithm for maintaining strong components under arc deletions, but no
additions, given an initial graph.  The expected amortized time per arc deletion is
$O(n)$ and the query time is $O(1)$, worst-case.  Their algorithm does not maintain a
topological order of the vertices but can be easily modified to do so, with the same
bounds.  If both additions and deletions are allowed, there is no known solution
better than running an $O(m)$-time static algorithm after each graph change, even for
the simplest problem, that of cycle detection.  There has been quite a bit of work on
the harder problem of maintaining full reachability information for a dynamic graph.
See \cite{Roditty2002,Roditty2004}.

For the incremental topological ordering problem there are a couple of lower bounds,
but there are large gaps between the existing lower bounds and the upper bounds.
Ramalingam and Reps \cite{Ramalingam1994} gave a class of examples in which $n - 1$
arc additions force $\Omega(n\log n)$ vertices to be reordered, no matter what
topological order is maintained.  This is the only general lower bound.  Katriel
\cite{Katriel2004} considered what she called the {\em topological sorting problem},
in which the topological order must be maintained as an explicit map between the
vertices and the integers between $1$ and $n$.  For algorithms that only renumber
vertices within the affected region, she gave a class of examples in which $O(n)$ arc
additions cause $\Omega(n^2)$ vertices to be reordered.  The algorithms of
Marchetti-Spaccamela et al., Pearce and Kelly, Ajwani et al., and the $O(n^{2.5})$
algorithm of Kavitha and Mathew are all subject to this bound, although our algorithm
is not.

For topological ordering algorithms that reorder only vertices within the affected
region, we can obtain a lower bound of $\Omega(nm^{1/2})$ on the total number of
vertex reorderings in the worst case, assuming $m = \Omega(n)$.  For simplicity
assume that $m$ is a perfect square and $m^{1/2}$ evenly divides $n$.  Number the
vertices from 1 through $n$ in their initial topological order; use these fixed
numbers to permanently identify the vertices even as the order changes.  Begin by
adding arcs to form $m^{1/2}$ paths, each consisting of a sequence of $k = n/m^{1/2}$
consecutive vertices.  Then add arcs $(2k, 1), (3k, 1),\ldots,(n, 1), (3k, k+1), (4k,
k+1),\ldots,(n, k+1), (4k, 2k+1), (5k, 2k+1),\ldots,(n, 2k+1),$ and so on. The total
number of arcs added is $n - m^{1/2}$ to form the initial paths plus $m^{1/2}(m^{1/2}
- 1)/2$, totaling less than $m$ if $n \le m/2 + 3m^{1/2}/2$.  After the initial
paths are formed, each arc addition forces the reordering of $k$ vertices if the
reordering is only within the affected region, for a total of $\Omega(nm^{1/2})$
vertex reorderings.  The reordering is forced: each arc addition causes two paths to
change places.  (See Figure \ref{fig:lb}.)

\begin{figure}
\centering
\includegraphics[scale=0.5]{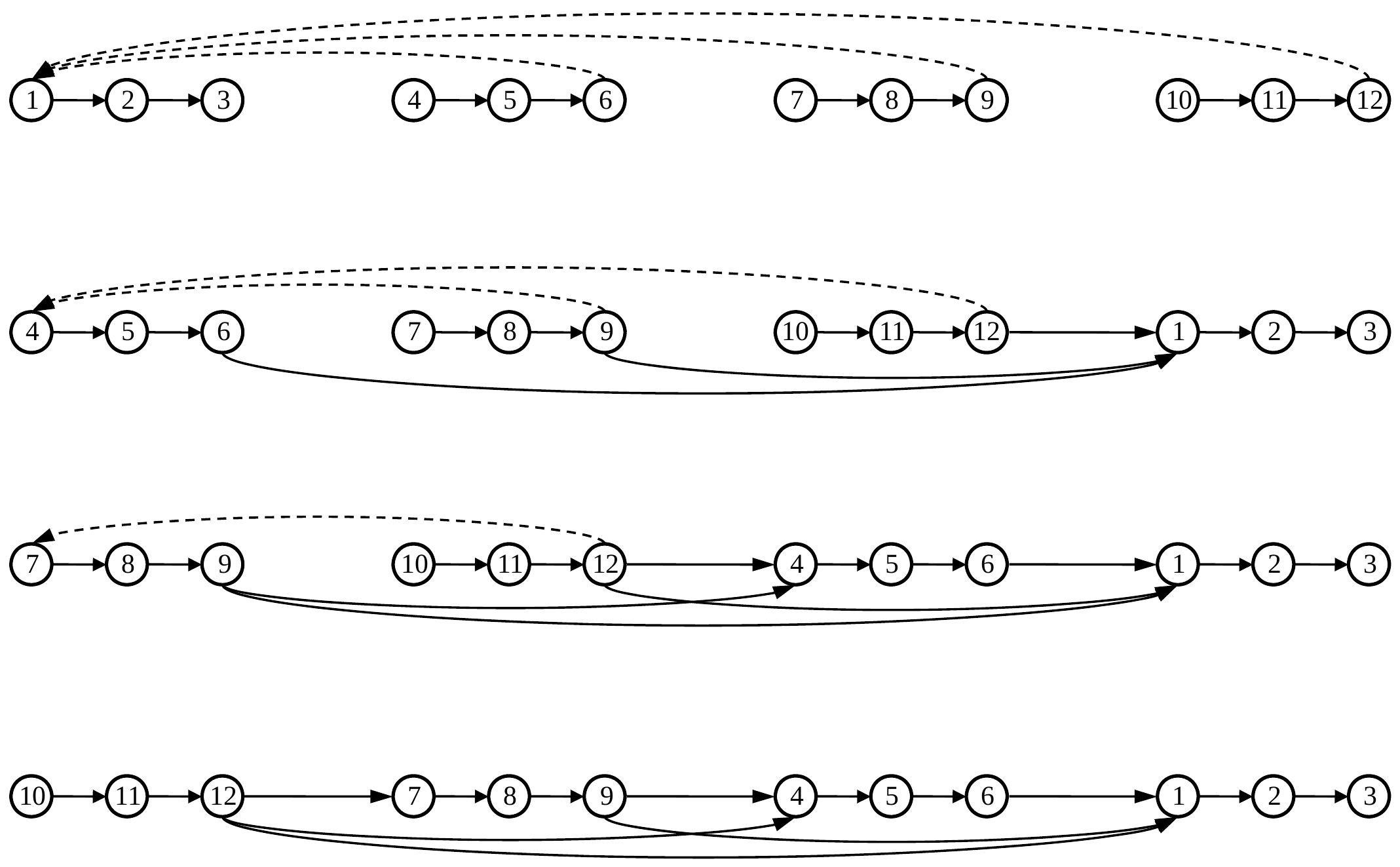} \\
\caption{The $\Omega(nm^{1/2})$ vertex reordering construction with $n = 12$, $m =
16$.}
\label{fig:lb}
\end{figure}

The bound of $\Omega(nm^{1/2})$ applies to all existing algorithms, including ours;
for sparse graphs, the running time of our algorithm matches this bound.  If we are
only interested in minimizing the number of vertex reorderings, not minimizing the
running time, we can get a matching upper bound of $O(nm^{1/2})$ on the number of
vertex reorderings by doing an ordered bidirectional search that alternates between
scanning a forward vertex and scanning a backward vertex; a count of related vertex
pairs gives the bound.

Another observation is that if the incident arc lists are sorted by end vertex, our
compatible search method can be modified so that the total search time over all arc
additions is $O(n^2)$: stop traversing arcs incident to a vertex when the next arc is
incident to a vertex outside the hard thresholds.  This bound, also, comes from a
count of related vertex pairs.  Unfortunately, keeping the arc lists sorted seems to
require more than $O(m^{3/2})$ time, giving us no actual improvement.  The
$O(n^{2.75})$-time algorithm of Ajwani et al. uses this idea but keeps the arc lists
partially sorted, trading off search time against arc list reordering time.

The running time analysis of Ajwani et al. and that of Kavitha and Mathew for their
$O(n^{2.5})$-time algorithm rely on a linear program to bound the total amount by
which vertex numbers change.  Although the solution to this linear program is
$\Theta(n^{2.5})$, it may not capture all the constraints of the problem, and Kavitha
and Mathew do not provide a class of examples for which their time bound is tight.
One would like such a class of examples, or alternatively a tighter analysis of their
algorithm.

We have used amortized running time as our measure of efficiency.  An alternative way
to measure efficiency is to use an incremental competitive model
\cite{Ramalingam1991}, in which the time spent to handle an arc addition is compared
against the minimum work that must be done by any algorithm, given the same current
topological order and the same arc addition.  The minimum work that must be done is
the minimum number of vertices that must be reordered, which is the measure that
Ramalingam and Reps used in their lower bound.  But no existing algorithm handles an
arc addition in time polynomial in the minimum number of vertices that must be
reordered.  To obtain positive results, some researchers have measured the
performance of their algorithms against the minimum sum of degrees of vertices that
must be reordered \cite{Alpern1990} or a more-refined measure that counts out-degrees
of forward vertices and in-degrees of backward vertices \cite{Pearce2006}.  For these
models, appropriately balanced forms of ordered search are competitive to within a
logarithmic factor \cite{Alpern1990,Pearce2006}.  In such a model, our algorithm is
competitive to within a constant factor.  We think, though, that such a model is
misleading: it does not account for the possibility that different algorithms may
maintain different topological orders, it does not account for correlated effects of
multiple arc additions, and good bounds have only been obtained for a model that may
overcharge the adversary.

Alpern et al. and Pearce and Kelly consider batched arc additions as well as single
arc additions.  We have not yet considered generalizing compatible search to handle
batched arc additions.  Doing so might lead improvements in practice, if not in
theory.

Our algorithm uses a vertex numbering scheme in which all vertices have distinct
numbers. Alpern et al. allowed vertices to be numbered the same if there is no path
between them, in an effort to minimize the number of distinct vertex numbers. Our
algorithm can be modified to include this idea, as follows. Add an extra level to the
dynamic list order structure: vertices are grouped into those of equal number; each
group is an element of a block in the doubly-linked list of blocks. (See Section
\ref{sec:impl}.) Start with all vertices in a single group. Having computed $X$ and
$Y$ after a search, delete all vertices in $X$ and $Y$ from their respective groups,
and delete all empty groups.  If $s$ is not forward, delete all forward vertices from
the group containing $s$ and add these vertices to $X$.  If $s$ is not backward,
delete all backward vertices from the group containing $s$ and add them to $Y$.  (If
$s$ is neither forward nor backward, do both.)  Assign each vertex in $X$ to a new
group corresponding to its maximum path length (in arcs) from $w$; assign each vertex
in $Y$ to a new group corresponding to its maximum path length (in arcs) to $v$.  If
$s$ is not forward, insert the new groups just after that of $s$, with the groups of
$Y$ first, in decreasing order of their maximum path length to $v$, followed by the
groups of $X$, in increasing order of their maximum path length from $w$.  Proceed
symmetrically if $s$ is not backward.  Computing maximum path lengths takes linear
time \cite{Tarjan1983}, so the overall time bound is unaffected.  This method extends
to the maintenance of strong components.  Whether this idea yields a speed-up in
practice is an experimental question.